\documentclass[10pt,twocolumn]{article}
\usepackage[letterpaper,margin=0.75in,columnsep=0.33in]{geometry}
\usepackage[T1]{fontenc}
\usepackage{amsmath,mathtools,amsthm}
\usepackage[tt=false]{libertine}
\usepackage[libertine]{newtxmath}
\usepackage{graphicx}
\usepackage{booktabs}
\usepackage{algorithm}
\usepackage{algpseudocode}
\usepackage{caption}
\usepackage{xcolor}
\usepackage[numbers,sort&compress]{natbib}
\usepackage[hidelinks]{hyperref}

\theoremstyle{plain}
\newtheorem{proposition}{Proposition}[section]
\theoremstyle{definition}
\newtheorem{definition}[proposition]{Definition}

\newcommand{\sysname}{\textsc{Lychee}}
\newcommand{\R}{\mathbb{R}}
\newcommand{\vq}{\mathbf{q}}
\newcommand{\vx}{\mathbf{x}}
\newcommand{\vc}{\mathbf{c}}

\newcommand{\lift}{\ell}
\title{\sysname: Certified Screening of Inverted Lists\\ for High-Recall Vector Search}
\author{Dongfang Zhao\\
University of Washington, Tacoma School of Engineering and Technology\\
\texttt{dzhao@uw.edu}}
\date{}

\begin{document}
\twocolumn[
\begin{@twocolumnfalse}
\maketitle
\begin{abstract}
Nearest-neighbor search over embedding vectors now matches faces against watchlists, transactions against known fraud
patterns, and medical images against confirmed cases, where a missed neighbor is a missed threat, fraud, or
diagnosis. Because an approximate index cannot report which queries lost a true neighbor, these applications require recall
of 0.999 to 0.9999, which cuts the misses of recall 0.99 by 10 to 100 times, at a cost far below that of exact search
over billions of vectors. Conventionally, inverted-file indexes reach high recall
by probing more lists. However, they read and score every probed vector at full precision, so that the throughput of
a query falls in inverse proportion to the number of probed lists. 
On the other hand, graph indexes are fast at moderate recall but
saturate below the high-recall targets, whereas quantized lists must rerank their candidates with the full vectors. Our
insight is to decide for almost every probed vector that it cannot enter the top $k$ of the probed lists without
reading the vector, through the certified lower bound of exact search by certified exclusion. This paper turns the
above insight into \sysname, a GPU inverted-file search that screens the probed lists on tensor cores under a
threshold from anchors of the same lists, and that computes exact distances only for the few vectors that survive
the screening, in increasing order of their certified bounds. We prove that \sysname\ returns exactly the $k$
nearest vectors of the probed lists, which is the answer of an inverted-file search at full precision, while it reads
a few hundred full vectors per query regardless of the number of probed lists. We evaluate \sysname\ on eleven
datasets from one million to one billion vectors against ten GPU and CPU baselines on two GPU servers. At equal
recall with 200 probed lists, \sysname\ is 9.7 to 117 times faster than IVF-Flat of cuVS on the nine datasets of
up to 100 million vectors, and 65 times faster on Deep-100M at recall 0.9997 with 1,000 probed lists. \sysname\ also reaches recall levels
that CAGRA and HNSW do not reach, and it outruns exact search by certified exclusion by 5 to 18 times on the
datasets of more than ten million vectors. On four GPUs, \sysname\ searches one billion vectors at recall
0.9999 and is 44.5 times faster than IVF-Flat on BIGANN-1B.
\end{abstract}
\vspace{1.5em}
\end{@twocolumnfalse}
]
\section{Introduction}
\label{sec:intro}

Nearest-neighbor search over embeddings increasingly decides outcomes in which a missed neighbor has a direct impact
on real-world applications. In security, face recognition~\cite{schroff2015facenet,deng2019arcface} searches a watchlist for the
nearest face embedding, so that a missed neighbor allows a wanted person to pass the check. In health care, retrieval of
medical images~\cite{muller2004cbir} supports a diagnosis with the confirmed cases nearest to the image of a
patient, so that a missed case withholds evidence from the physician. In finance, fraud and anti-money-laundering
screening matches transactions against known fraud patterns and sanctioned parties, where a missed match is both a
loss and a compliance violation. Legal discovery~\cite{grossman2011tar}, the deduplication of training
data~\cite{lee2022dedup,abbas2023semdedup}, and retrieval-augmented generation~\cite{lewis2020rag} over dense
passage retrieval~\cite{karpukhin2020dpr,hofstatter2021tasb,khattab2020colbert} similarly need every relevant
document, near-duplicate, or passage.

These applications need recall far above the 0.9 to 0.95 at which approximate indexes are usually compared, because
the errors of an approximate index are silent. An index at recall 0.99 loses about one true neighbor in a hundred
without reporting which queries lost one, so that no operator can detect or repair the miss. A screening service
that answers one million identification queries per day at recall 0.99 therefore misses about ten thousand true
matches per day. Since the misses scale with one minus the recall, recall 0.9999 reduces them to about one hundred.
Ranking quality follows the same trend: on MS~MARCO~\cite{nguyen2016msmarco}, the mean reciprocal rank (MRR@10) of the retrieved passages rises from 0.3295 at
recall 0.957 to 0.3373 at recall 0.997, against 0.3378 for the exact neighbors (Section~\ref{sec:msmarco}). Exact
search with CertEx~\cite{zhao2026certex} removes the misses, but its cost grows in proportion to the dataset, whereas the
embedding datasets of these applications, such as DINOv2~\cite{oquab2024dinov2}, CLIP~\cite{radford2021clip}, and
DataComp~\cite{gadre2023datacomp}, hold millions to billions of vectors. These applications therefore need recall of
0.999 to 0.9999 at the cost of an approximate index, which lies at the expensive end of its recall and throughput
trade-off.

The inverted-file (IVF) index~\cite{sivic2003videogoogle,jegou2011pq},
as used in
Faiss~\cite{johnson2019faiss,douze2024faisslib} and cuVS~\cite{nvidia2025cuvs},
is the standard partition index for
this trade-off, which splits the dataset into lists by $k$-means. A query then scans the vectors of its $n_p$ nearest
lists, whose recall grows with $n_p$ because the true neighbors of a query near a list boundary lie in neighboring
lists.
However, IVF-Flat reads and scores every probed vector at full precision. On ImageNet DINOv2 embeddings, the
throughput of IVF-Flat in cuVS falls from 562,367 to 3,450 queries per second (QPS) as $n_p$ grows from 1 to 200 of
1,024 lists (Figure~\ref{fig:opener}a), while recall at 10 rises from 0.860 to 0.999. The other families of indexes
respond to this cost in two ways. First, compressed lists (IVF-PQ~\cite{jegou2011pq}, ScaNN~\cite{guo2020scann},
RaBitQ~\cite{gao2024rabitq}) shrink the scanned bytes but must rerank their candidates with the full vectors to
reach high recall. Second, graph indexes such as HNSW~\cite{malkov2020hnsw} and CAGRA~\cite{ootomo2024cagra} visit
few vectors per query but saturate below the targets: CAGRA reaches recall 0.986 at its largest search width on
DINOv2.

\begin{figure}[t]
  \centering
  \includegraphics[width=\columnwidth]{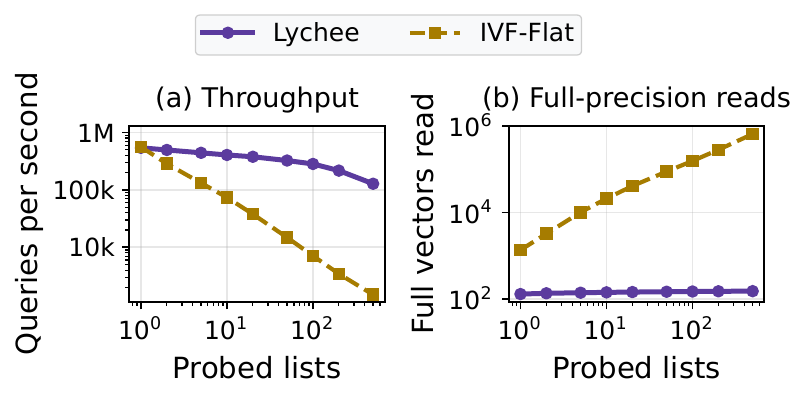}
  \caption{Recall and throughput on DINOv2 embeddings.}
  \label{fig:opener}
\end{figure}

One observation is that the probed lists of a high-recall query hold far more vectors than its answer needs. 
For instance, at $n_p = 200$, a DINOv2 query
probes about 280,000 vectors. Almost all of them lie far outside the radius of the 10th neighbor, so that the exact
top 10 among them can be certified from the full distances of fewer than 200 vectors (Figure~\ref{fig:opener}b). A
full-precision scan nevertheless spends its memory traffic and arithmetic on the far vectors. 
The cost of high recall
in an IVF index is therefore the cost of proving that far vectors are far.

Our insight is to certify the exclusion of far vectors inside the probed lists from low-dimensional projections,
rather than reading and scoring their full vectors. While exact search by certified exclusion
(CertEx)~\cite{zhao2026certex} proves that distant vectors cannot belong to the global nearest neighbors through
hardware-certified lower bounds on short projection rows, an IVF search requires certifying non-membership
within an irregular, query-dependent subset of lists. Proving exclusion inside the probed lists requires bounding
the $k$-th distance local to these lists from their own anchors, which enables tensor cores to screen candidate rows
at a fraction of full-precision memory traffic. Full-precision reads and exact distances are then needed only for the
few survivors whose bounds can still enter the probed top $k$, breaking the linear dependence of high-recall
IVF search on the number of probed vectors.

This paper turns the above insight into \sysname, a GPU IVF search with certified screening of the probed
lists. \sysname\ takes its threshold from exact distances to anchors of the probed lists, which bound the $k$-th
distance within the lists from above. A routed tensor-core pass screens all (query, list) pairs of a batch at once,
grouping the queries that probe a list so that every list streams through the tensor cores once per batch. The
survivors of all lists of a query are then verified together in increasing order of their certified bounds, until
no remaining survivor can enter the top $k$. We prove that this procedure returns exactly the $k$ nearest vectors of
the probed lists, namely the answer of IVF-Flat on the same lists.

\sysname\ changes the cost of high recall in an IVF index from the number of probed vectors to the number of
survivors of the certified screening. On DINOv2, \sysname\ serves 217,195 QPS at recall 0.999 ($n_p = 200$), 63
times the throughput of IVF-Flat and 53 times that of IVF-PQ with refinement at the same recall. The evaluation spans eleven datasets from one
million to one billion vectors, ten baselines that cover the inverted-file, graph, and exhaustive families on GPUs
and CPUs, and two GPU servers with five H100 GPUs in total. 

In summary, this paper makes the following contributions:
\begin{itemize}
  \item We propose a list-local certified threshold for inverted-file search, which bounds the $k$-th distance
        within the probed lists from the exact distances and thereby admits the certified
        screening of excluded vectors inside the lists.
  \item We prove that certified screening followed by tie-aware ordered verification returns exactly the $k$
        nearest vectors of the probed lists and derive a cost model in which the full-precision reads of a query
        follow its survivors rather than its probed vectors.
  \item We implement \sysname\ on GPUs with a routed tensor-core pass that screens all (query, list) pairs of a
        batch at once and verifies the survivors of all probed lists of a query together.
  \item We evaluate \sysname\ on eleven datasets from one million to one billion vectors against ten GPU and
        CPU baselines on two GPU servers and show that \sysname\ is 9.7 to 117 times faster than IVF-Flat at
        equal recall with 200 probed lists and up to 44.5 times faster on one billion vectors on four GPUs.
\end{itemize}

\section{Related Work}
\label{sec:related}

\paragraph{Partition and quantization indexes.} Inverted lists of a coarse
quantizer~\cite{sivic2003videogoogle,jegou2011pq} are the basis of the IVF indexes in
Faiss~\cite{johnson2019faiss,douze2024faisslib}, cuVS~\cite{nvidia2025cuvs}, and vector databases such as
Milvus~\cite{wang2021milvus}, Manu~\cite{guo2022manu}, AnalyticDB-V~\cite{wei2020analyticdbv}, and
VBase~\cite{zhang2023vbase}. The inverted multi-index~\cite{babenko2012imi}, revisited inverted
indices~\cite{baranchuk2018revisiting}, and SPANN~\cite{chen2021spann} improve the partition and the probe order,
while product quantization~\cite{jegou2011pq}, optimized product quantization~\cite{ge2014opq}, additive
quantization~\cite{babenko2014additivequant}, ScaNN~\cite{guo2020scann}, RaBitQ~\cite{gao2024rabitq}, its extended
variant~\cite{gao2025extrabitq}, and IVF-RaBitQ on GPUs~\cite{shi2026ivfrabitq} compress the scanned vectors and
rerank with full vectors~\cite{jegou2011bigann}. These methods change which vectors are probed or how many bytes an
approximate distance reads, whereas \sysname\ keeps the probed set of IVF-Flat and certifies the exclusion of probed
vectors, so that its answer equals the exact top $k$ of the probed lists.

\paragraph{Graph indexes.} Graph indexes such as NSG~\cite{fu2019nsg}, HNSW~\cite{malkov2020hnsw},
DiskANN~\cite{subramanya2019diskann}, Starling~\cite{wang2024starling}, and their filtered
variants~\cite{golub2024filtereddiskann,patel2024acorn} visit few vectors per query through greedy search on a
neighborhood graph built by methods such as NN-Descent~\cite{dong2011nndescent}. On GPUs,
SONG~\cite{zhao2020song}, GGNN~\cite{groh2023ggnn}, CAGRA~\cite{ootomo2024cagra}, BANG~\cite{venkatasubba2025bang},
and PilotANN~\cite{gui2026pilotann} parallelize the search. Surveys~\cite{wang2021annsurvey} and
benchmarks~\cite{aumuller2020annbench,simhadri2022bigann} show that graph searches reach high throughput at moderate
recall but saturate as the search width grows. \sysname\ instead reaches the recall of the inverted lists at every $n_p$, up to 0.9999 on DINOv2,
with a cost that grows with the survivors of the screening.

\paragraph{Exact search and lower bounds.} Lower-bounding filters for exact search date back to the GEMINI
framework~\cite{faloutsos1994gemini}, the VA-file~\cite{weber1998vafile}, metric trees such as
AESA~\cite{vidal1986aesa}, vp-trees~\cite{yianilos1993vptree}, M-trees~\cite{ciaccia1997mtree}, and cover
trees~\cite{beygelzimer2006cover}, and iDistance~\cite{jagadish2005idistance}, which exclude candidates whose lower
bound exceeds a current threshold. In high dimensions these bounds lose their
power~\cite{beyer1999meaningful,weber1998vafile}, so that GPU exact search falls back to brute
force~\cite{garcia2008gpuknn,johnson2019faiss}. CertEx~\cite{zhao2026certex}
restores the power of the lower bound in high dimensions through a principal-component lift with a rounding
certificate for tensor cores. \sysname\ applies this bound inside the probed lists of an IVF index under a
list-local threshold, which yields the answer of IVF-Flat on the same lists.

\begin{figure*}[t]
  \centering
  \includegraphics[width=\textwidth]{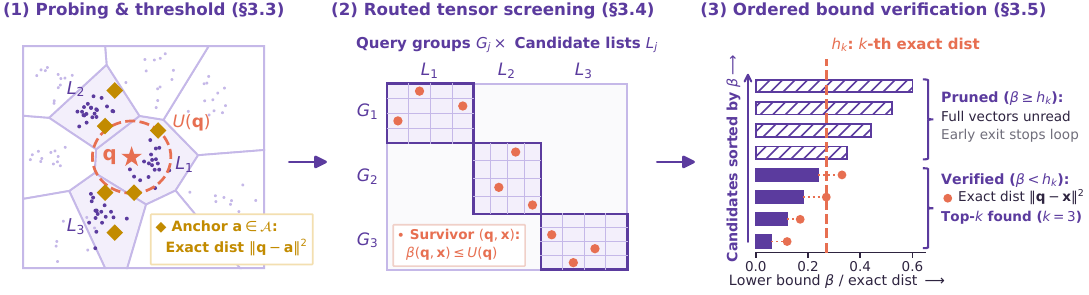}
  \caption{The search pipeline of \sysname\ for a batch of queries across its three execution stages.}
  \label{fig:pipeline}
\end{figure*}

\section{Methodology}
\label{sec:method}

\subsection{Preliminaries and Background}
\label{sec:background}

\paragraph{Inverted-file search.} An inverted-file index~\cite{sivic2003videogoogle,jegou2011pq} partitions a dataset $\mathcal{X} = \{\vx_1, \ldots, \vx_N\} \subset \R^D$ into $P$ lists
$L_1, \ldots, L_P$ around centers $\vc_1, \ldots, \vc_P$ from $k$-means, assigning every vector to its nearest
center. A query $\vq$ ranks the centers by distance and probes the $n_p$
nearest ones. The probed lists form the candidate set $S(\vq)$, from which IVF-Flat returns the $k$ vectors that
lie nearest to $\vq$ under the squared Euclidean distance. Recall therefore depends on the lists and on $n_p$ alone,
because the scan of $S(\vq)$ is exact up to floating-point rounding. Multi-index
structures~\cite{babenko2012imi}, learned routing~\cite{baranchuk2018revisiting}, and balanced partitions for
disk-resident data~\cite{chen2021spann} refine the choice of $S(\vq)$, although each of them still reads the vectors
of the chosen lists.
The scan of the probed lists dominates the cost of a high-recall query. A query reads $|S(\vq)| \approx n_p N / P$
vectors of $4D$ bytes each, so that its memory traffic grows linearly with $n_p$. GPU implementations such as
Faiss~\cite{johnson2019faiss} and cuVS~\cite{nvidia2025cuvs} therefore spend the time of a high-recall query on
this traffic. Compressed
lists reduce the bytes per vector through product quantization~\cite{jegou2011pq,ge2014opq}, anisotropic
quantization~\cite{guo2020scann}, or randomized binary codes~\cite{gao2024rabitq,shi2026ivfrabitq}. However, the
compressed distances reorder the candidates, so that high recall requires a refinement step~\cite{jegou2011bigann} that reads the full
vectors of a multiple of $k$ candidates per query. The cost of a high-recall IVF query is
therefore set by its full-precision reads, either of the whole probed set or of the refined candidates.

\paragraph{Certified exclusion.} Recently, certified exclusion (CertEx)~\cite{zhao2026certex} was proposed to replace most full-precision reads with a certified test. Consider an orthonormal matrix $V \in \R^{D \times m}$ of the $m$ principal directions
of the dataset and its mean $\boldsymbol{\mu}$. CertEx lifts every vector to
\begin{equation}
  \lift(\vx) = \bigl[\, V^\top (\vx - \boldsymbol{\mu});\ \|(I - V V^\top)(\vx - \boldsymbol{\mu})\| \,\bigr] \in \R^{m+1},
  \label{eq:lift}
\end{equation}
The Cauchy-Schwarz inequality on the residual components then gives $\|\lift(\vq) - \lift(\vx)\| \le \|\vq - \vx\|$
for every pair. The lifted rows are stored in half precision together with four certificate coordinates that bound
their rounding errors, so that one tensor-core product of a query row with a vector row yields a lower bound
$\beta(\vq, \vx) \le \|\vq - \vx\|^2$ that holds under every rounding of the product.
A vector $\vx$ survives the screening of a query only if $\beta(\vq, \vx) \le U(\vq)$, where $U(\vq)$ bounds the
squared distance of the $k$-th nearest neighbor from above. CertEx sets $U(\vq)$ from the exact distances of sampled
anchors, since any $k$ distinct vectors bound the $k$-th distance, and verifies the survivors in increasing order of
$\beta$ until the next bound exceeds the $k$-th exact distance found. Every vector that the screening excludes costs
$2(m + 5)$ bytes instead of $4D$, because only the survivors are read at full precision. CertEx applies this test to the
whole dataset, so that its exact search still screens all $N$ rows per query. In contrast, an inverted-file query
needs only the $k$ nearest vectors of its probed lists, which requires a threshold that bounds the $k$-th distance
within these lists.

\subsection{\sysname\ Overview}
\label{sec:problem}

\sysname\ performs approximate nearest-neighbor search by returning the exact top-$k$ candidates of the probed lists, matching the full-precision candidate ranking of IVF-Flat. Given a query $\vq \in \R^D$ and probed set $S(\vq) = L_{j_1} \cup \cdots \cup L_{j_{n_p}}$, the target answer $A_k(\vq)$ minimizes the sum of squared distances over any $k$-subset:
\begin{equation}
  A_k(\vq) = \operatorname*{arg\,min}_{T \subseteq S(\vq),\, |T| = k} \sum_{\vx \in T} \|\vq - \vx\|^2 ,
  \label{eq:target}
\end{equation}
with ties broken by index. Evaluating Eq.~\eqref{eq:target} exactly attains the recall of the inverted lists at every $n_p$, leaving query latency and throughput as the only quantities to improve.

The cost of Eq.~\eqref{eq:target} is dominated by full-precision reads. Define $r_k(\vq)$ as the $k$-th smallest
squared distance from $\vq$ to $S(\vq)$. Every vector of $S(\vq)$ farther than $r_k(\vq)$ is excluded from
$A_k(\vq)$, although a full-precision scan learns this fact only after reading the vector. \sysname\ instead proves
exclusion from the low-dimensional lifted row of the vector, under a threshold $U(\vq) \ge r_k(\vq)$ that the probed
lists themselves provide.

Figure~\ref{fig:pipeline} illustrates the resulting search pipeline of a query batch across its three algorithmic
stages. In Stage 1, the search identifies $n_p$ probed lists and gathers exact distances to their anchors (Definition~\ref{def:anchors}) to
establish a list-local threshold $U(\vq) \ge r_k(\vq)$ (Figure~\ref{fig:pipeline}, panel~1). In Stage 2, a routed
tensor-core pass groups queries probing the same candidate list $L_j$ into a query group $G_j$, screening their half-precision lifted rows on tensor cores and filtering out candidates whose certified lower bound $\beta(\vq, \vx)$ exceeds $U(\vq)$ (Figure~\ref{fig:pipeline}, panel~2). In
Stage 3, the survivors of all probed lists are aggregated per query and verified at full precision in increasing
order of their certified bounds, terminating verification as soon as the lower bound reaches the current $k$-th
exact distance $h_k$ (Figure~\ref{fig:pipeline}, panel~3).

\subsection{List-Local Certified Threshold}
\label{sec:threshold}

The screening threshold of \sysname\ is established from exact distances to a small set of representative vectors chosen from the probed lists, termed \emph{anchors} (Figure~\ref{fig:pipeline}, panel~1).

\begin{definition}[List Anchors and Probed Anchor Set]
\label{def:anchors}
Consider an inverted list $L_j = (\vx_{j, 1}, \ldots, \vx_{j, |L_j|}) \subset \R^D$. For a fixed stride $s \in \mathbb{N}^+$, the \emph{anchors} of $L_j$ are the designated subset of stored vectors at regular index intervals:
\begin{equation}
  \Omega(L_j) = \bigl\{\, \vx_{j, i \cdot s} \;\big|\; 1 \le i \le \lfloor |L_j| / s \rfloor \,\bigr\} \subset L_j,
\end{equation}
whose full-precision representations are preserved in device memory. For a query $\vq$ probing lists $S(\vq) = \bigcup_{j \in J(\vq)} L_j$, where $J(\vq) = (j_1, \ldots, j_{n_p})$ ranks cluster centers by distance $\|\vq - \vc_{j_1}\| \le \cdots \le \|\vq - \vc_{j_{n_p}}\|$, the \emph{probed anchor set} $\mathcal{A}(\vq) \subseteq S(\vq)$ is the sequence of the first $b$ anchors gathered from $\Omega(L_{j_1}), \Omega(L_{j_2}), \ldots$ in center order, where $b \ge k$ denotes the anchor budget.
\end{definition}

Given $\mathcal{A}(\vq)$, \sysname\ computes the exact squared distance $\|\vq - \vx\|^2$ for each anchor $\vx \in \mathcal{A}(\vq)$. The list-local certified threshold is then
\begin{equation}
  U(\vq) = (1 + \tau) \, d_{(k)}\bigl(\vq, \mathcal{A}(\vq)\bigr),
  \label{eq:threshold}
\end{equation}
where $d_{(k)}$ is the $k$-th smallest exact squared distance among the anchors $\mathcal{A}(\vq)$, and $\tau$ is a small relative slack that covers the rounding of the screening threshold. If the probed lists hold fewer than $k$ anchors, $U(\vq)$ falls back to a bound that every stored vector satisfies, computed from the norm of the query and the largest norm of the stored vectors. The query then scans its fewer than $(k + n_p)\, s$ probed vectors exactly.

\begin{proposition}
\label{prop:threshold}
Under Definition~\ref{def:anchors}, if~$|\mathcal{A}(\vq)| \ge k$, then $U(\vq) \ge r_k(\vq)$.
\end{proposition}

\begin{proof}
By Definition~\ref{def:anchors}, the anchors are $|\mathcal{A}(\vq)| \ge k$ distinct vectors of $S(\vq)$. The $k$ nearest of them form a subset of
$S(\vq)$ of size $k$ whose largest squared distance is $d_{(k)}$. The $k$-th smallest squared distance over the
superset $S(\vq)$ cannot exceed the largest distance of any of its subsets of size $k$, so that $r_k(\vq) \le d_{(k)}
\le U(\vq)$.
\end{proof}

The threshold of Proposition~\ref{prop:threshold} is local to the probed lists. A threshold from anchors of the
whole dataset bounds the $k$-th distance in the dataset, which is smaller than $r_k(\vq)$ when the probed lists
miss some true neighbors. A screening under such a threshold could therefore return fewer than $k$ vectors of
$S(\vq)$. The list-local threshold instead keeps every member of $A_k(\vq)$ under the screening at every $n_p$.
Moreover, the anchors of the nearest lists come first, so that the budget $b$ spends exact distances on the anchors
that are most likely to lie near $\vq$ and keeps $U(\vq)$ close to $r_k(\vq)$.

\subsection{Routed Tensor Screening}
\label{sec:routed}

Given the list-local threshold $U(\vq)$ of each query from Stage 1, Stage 2 screens candidate vectors within probed lists on GPU tensor cores. A naive per-query scan cannot map onto tensor cores efficiently: each query probes a different subset of $n_p$ lists, while an individual inverted list of a few thousand vectors offers too little work for a single query to fill a tensor-core tile. \sysname\ therefore batches $B$ queries and groups the queries that probe list $L_j$ into a query group $G_j$. The queries in $G_j$ and the candidate vectors in $L_j$ then form dense matrix tiles of $128 \times 128$ rows on tensor cores (Figure~\ref{fig:pipeline}, panel~2). Every list therefore streams through cache and tensor cores once per batch, however many queries probe it. All lists share the global basis $V$ of CertEx, so that one lifted row of a query screens all its probed lists; Section~\ref{sec:work} shows that this basis excludes all but a few hundred probed vectors.

For each query row and candidate row in a tile, tensor cores compute their half-precision inner product, incorporating the certificate coordinates of Eq.~\eqref{eq:lift} to establish the certified lower bound $\beta(\vq, \vx) \le \|\vq - \vx\|^2$ under finite-precision rounding~\cite{zhao2026certex}. An in-register predicate compares $\beta(\vq, \vx)$ against the threshold $U(\vq)$ of that query, retaining vector $\vx$ only if $\beta(\vq, \vx) \le U(\vq)$. The vectors of list $L_j$ satisfying this predicate form the list survivors, depicted as coral dots in Figure~\ref{fig:pipeline} (panel~2). Aggregating these survivors over all $n_p$ probed lists of query $\vq$ produces the complete candidate survivor set:
\begin{equation}
  C \gets \bigl\{\vx \in S(\vq) \;\big|\; \beta(\vq, \vx) \le U(\vq)\bigr\},
  \label{eq:survivors}
\end{equation}
which is forwarded to Stage 3.

The following proposition establishes that certified screening prunes unpromising vectors with zero false dismissals.

\begin{proposition}[Screening Soundness]
\label{prop:screening}
Consider the exact top-$k$ nearest neighbors $A_k(\vq) \subseteq S(\vq)$ of query $\vq$ in its probed lists (Eq.~\eqref{eq:target}). If the list-local threshold satisfies $U(\vq) \ge r_k(\vq)$, then every true neighbor survives the certified screening: $A_k(\vq) \subseteq C$.
\end{proposition}
\begin{proof}
Consider any vector $\vx^* \in A_k(\vq)$. By definition of the $k$-th distance, $\|\vq - \vx^*\|^2 \le r_k(\vq)$. Because the certified lower bound is sound under finite-precision arithmetic, $\beta(\vq, \vx^*) \le \|\vq - \vx^*\|^2$. By transitivity and the validity of the threshold, $\beta(\vq, \vx^*) \le r_k(\vq) \le U(\vq)$. Vector $\vx^*$ therefore satisfies the screening predicate, so $\vx^* \in C$.
\end{proof}

\subsection{Ordered Bound Verification}
\label{sec:verification}

In Stage 3, \sysname\ verifies the survivors of the screening in increasing order of their certified lower bounds
$\beta(\vq, \vx)$, computing exact squared distances $\|\vq - \vx\|^2$ at full precision. Algorithm~\ref{alg:search}
summarizes the search of a query across the three stages.

Lines 5 to 10 execute the ordered verification over the survivor set $C$. Verification stops as soon as the next lower
bound, raised to 0 since a squared distance is never negative, reaches the current $k$-th exact distance $h_k$
(Figure~\ref{fig:pipeline}, panel~3). Every remaining survivor in $C$ then has an exact distance at least as large, so
that it can at best tie with the $k$-th neighbor already found. Stopping verification at equality is critical for
datasets with duplicate vectors (such as DataComp), in which thousands of candidates share the $k$-th distance of a
query and would otherwise all be verified. On the GPU, the verification runs with one warp per query computing exact
distances in float64.

\begin{algorithm}[t]
\small
\caption{Certified screening of the probed lists of a query.}
\label{alg:search}
\begin{algorithmic}[1]
\Require query $\vq$, number of probed lists $n_p$, anchor budget $b$
\Ensure $A_k(\vq)$, the $k$ nearest vectors of the probed lists
\State $J \gets$ the $n_p$ centers nearest to $\vq$; \ $S \gets \bigcup_{j \in J} L_j$
\State $\mathcal{A} \gets$ the first $b$ anchors of $S$ in the order of $J$ \Comment{Def.~\ref{def:anchors}}
\State $U \gets (1 + \tau)\, d_{(k)}(\vq, \mathcal{A})$ \Comment{Eq.~\eqref{eq:threshold}}
\State $C \gets \{\vx \in S : \beta(\vq, \vx) \le U\}$ \Comment{certified screening}
\State $H \gets \emptyset$ \Comment{top $k$ so far; $h_k$: its $k$-th distance}
\For{$\vx \in C$ in increasing order of $\beta(\vq, \vx)$}
  \If{$|H| = k$ \textbf{and} $\max(\beta(\vq, \vx), 0) \ge h_k$} \textbf{break}
  \EndIf
  \State insert $\vx$ into $H$ with distance $\|\vq - \vx\|^2$, keeping the $k$ nearest
\EndFor
\State \Return $H$
\end{algorithmic}
\end{algorithm}

\begin{proposition}
\label{prop:exact}
Algorithm~\ref{alg:search} returns $\min(k, |S(\vq)|)$ vectors of $S(\vq)$ whose squared distances to $\vq$ equal
those of $A_k(\vq)$.
\end{proposition}

\begin{proof}
Consider any $\vx \in A_k(\vq)$. Its squared distance is at most $r_k(\vq)$, which is at most $U$ by
Proposition~\ref{prop:threshold}. The certified bound satisfies $\beta(\vq, \vx) \le \|\vq - \vx\|^2 \le U$, so
that $\vx$ belongs to $C$. The loop over $C$ visits the survivors in increasing order of $\beta$ and stops at a
survivor $\vx'$ only if $|H| = k$ and $\max(\beta(\vq, \vx'), 0)$ reaches the $k$-th distance in $H$. Every
survivor that follows $\vx'$ in the order has a bound at least $\beta(\vq, \vx')$ and a nonnegative distance, so
that its distance is at least the $k$-th distance in $H$. The distances in $H$ are therefore the $k$ smallest
distances over $C$. Since $A_k(\vq) \subseteq C \subseteq S(\vq)$, the $k$ smallest distances over $C$ are those
over $S(\vq)$, namely the distances of $A_k(\vq)$.
\end{proof}

\subsection{Cost Analysis}
\label{sec:cost}

Algorithm~\ref{alg:search} runs in time $O(|S(\vq)|(m + 5) + b D + |C| \log |C| + V(\vq) D)$ per query, where $V(\vq)$
counts the survivors verified in Lines 6 to 10. The four terms account for the certified screening on tensor cores,
the exact distances of the anchors, the sort of the survivors, and the verification at full precision, compared to
$O(|S(\vq)| D)$ for IVF-Flat. The extra space per query is $O(b + |C|)$.

The memory traffic of \sysname\ depends on the probed rows primarily through the short lifted rows ($2(m + 5)$ bytes
each). Full-precision $4D$-byte reads are restricted to the $b$ anchors and $V(\vq)$ verified survivors. Against the
$4D\,|S(\vq)|$ bytes read by IVF-Flat, the per-query traffic ratio approaches $2D / (m + 5)$ as $n_p$ grows (e.g.,
$3\times$ for DINOv2 at $m = 507$). The routed pass of Section~\ref{sec:routed} divides this traffic further, because
a probed list streams through the tensor cores once for its whole query group $G_j$, whereas IVF-Flat reads the list
once for every query that probes it. The reuse across the batch therefore accounts for the speedup on 8-bit descriptors,
whose lifted rows are longer than the vectors (256 against 128 bytes for SIFT at $m = 123$). Device memory per vector is $2(m + 5) + 4D$ bytes when full vectors reside on the GPU and $2(m + 5) +
4D/s$ bytes when they reside in pinned host memory, since only the anchors stay on the GPU at full precision. Finally,
while CertEx~\cite{zhao2026certex} screens all $N$ lifted rows of the dataset, \sysname\ screens only the probed set
$|S(\vq)| \approx n_p N / P$, scaling down the screening work by $n_p / P$ (e.g., 9.8\% of the lists on SIFT1M and 3.1\%
on BIGANN-100M at recall 0.999).

\section{Evaluation}
\label{sec:eval}

\subsection{Experimental Setup}
\label{sec:setup}

\paragraph{Datasets.} The evaluation covers eleven datasets of text and image embeddings and local descriptors, from one
million to one billion vectors and from 96 to 960 dimensions, which test the certified screening on wide and narrow
vectors at every scale (Table~\ref{tab:datasets}). ImageNet DINOv2 holds the DINOv2 ViT-B/14
embeddings~\cite{oquab2024dinov2} of the training images of ImageNet~\cite{deng2009imagenet,russakovsky2015ilsvrc},
queried by its validation images. GIST1M holds the 960-dimensional GIST descriptors~\cite{oliva2001gist} of the
ANN-Benchmarks suite~\cite{aumuller2020annbench}. DataComp holds the CLIP ViT-L/14 embeddings~\cite{radford2021clip}
of 12.79 million web images of DataComp~\cite{gadre2023datacomp}, which contain large groups of duplicate images.
SIFT1M and the BIGANN prefixes hold the SIFT descriptors~\cite{lowe2004sift} of the BIGANN
dataset~\cite{jegou2011bigann} as 8-bit integers, the Deep1B prefixes hold the 96-dimensional image descriptors
of Deep1B~\cite{babenko2016deep}, and MS MARCO holds the 768-dimensional embeddings of its 8.8 million
passages~\cite{nguyen2016msmarco} by msmarco-distilbert-cos-v5~\cite{reimers2019sbert}, normalized so that Euclidean
distance ranks as cosine similarity, with the 6,980 labeled development queries. Ground truth is the exact top 10, which an exhaustive float64 search provides for
the Deep1B prefixes and corrects for three queries of the official Deep1B list.

\begin{table}[t]
  \centering
  \caption{The eleven datasets of the evaluation.}
  \label{tab:datasets}
  \footnotesize
  \begin{tabular*}{\columnwidth}{@{\extracolsep{\fill}}lrrlrr@{}}
    \toprule
    Dataset & Vectors & Dim. & Type & Queries & Lists \\
    \midrule
    ImageNet DINOv2 & 1.28M & 768 & float32 & 50,000 & 1,024 \\
    GIST1M & 1M & 960 & float32 & 1,000 & 1,024 \\
    SIFT1M & 1M & 128 & float32 & 10,000 & 1,024 \\
    Deep-10M & 10M & 96 & float32 & 10,000 & 4,096 \\
    BIGANN-10M & 10M & 128 & uint8 & 10,000 & 4,096 \\
    MS MARCO & 8.8M & 768 & float32 & 6,980 & 4,096 \\
    DataComp & 12.8M & 768 & float32 & 10,000 & 4,096 \\
    Deep-100M & 100M & 96 & float32 & 10,000 & 16,384 \\
    BIGANN-100M & 100M & 128 & uint8 & 10,000 & 16,384 \\
    BIGANN-1B & 1B & 128 & uint8 & 10,000 & $4 \times 16{,}384$ \\
    Deep1B & 1B & 96 & float32 & 10,000 & $4 \times 16{,}384$ \\
    \bottomrule
  \end{tabular*}
\end{table}

\paragraph{Platforms.} The experiments run on two GPU servers with the same GPU model and software stack, namely CUDA
12.6 with driver 560.35, PyTorch 2.14, and cuVS 26.10. The first server has one NVIDIA H100 GPU~\cite{nvidia2022h100} with 94~GB of
memory, 48 cores of an Intel Xeon Platinum 8468 processor, and 240~GB of host memory. It runs
\sysname, IVF-Flat, IVF-PQ, and CAGRA on ImageNet DINOv2, SIFT1M, BIGANN-10M, BIGANN-100M, and DataComp, together
with the width and list studies on ImageNet DINOv2 and SIFT1M. The second server has four H100 GPUs of the same
model, which NVLink connects pairwise, 192 cores of the same processor, and 960~GB of host memory. It runs every
method on GIST1M, Deep-10M, and Deep-100M, the remaining baselines on the other datasets, and the billion-scale
experiments on all four GPUs; every single-GPU experiment uses one of its GPUs. Every search processes all queries of a
dataset at once.

\paragraph{Metrics.} Recall at 10 is the fraction of the true 10 nearest neighbors in the returned list, averaged over
the queries. On DataComp, whose duplicate images make the ids of tied neighbors interchangeable, a returned vector
counts as a true neighbor if its distance is at most the 10th true distance, as in standard benchmarks with
ties~\cite{simhadri2022bigann}. Throughput is the number of queries divided by the wall-clock time of the search,
including the probe, the screening, and the verification, as the best of two or three runs after one warm-up run. The throughput of one configuration differs by up to 4\% between the separately timed runs of different sections.
Exactness is checked separately: on every dataset, at 1 to 128 probed lists, every answer of \sysname\ equals the
float64 top 10 of its probed lists, with the same ids or the same distances under ties, for 1,000 queries per
dataset of up to ten million vectors and 500 per larger dataset.

\paragraph{Baselines.} The evaluation compares \sysname\ against ten baselines that cover the inverted-file, graph,
and exhaustive families on GPUs and CPUs (Table~\ref{tab:baselines}). IVF-Flat of cuVS~\cite{nvidia2025cuvs} scans the probed lists at full precision and computes the answer of
Eq.~\eqref{eq:target} in float32 on its own $k$-means lists. IVF-PQ compresses the vectors to $D/2$ bytes of 8-bit
codes, and IVF-SQ quantizes every coordinate to 8 bits; both refine their top $4k$ candidates with the full vectors.
IVF-RaBitQ~\cite{gao2024rabitq,gao2025extrabitq,shi2026ivfrabitq} stores 7-bit extended RaBitQ codes and runs with
and without the same refinement; it failed to build on DataComp. CAGRA~\cite{ootomo2024cagra} builds a graph of
degree 64 and sweeps its search width from 16 to 512, on every dataset except DataComp, whose duplicate groups made
the CAGRA build fail. The exhaustive search of cuVS computes all float32 distances. On the CPU, the IVF-Flat index of
Faiss~\cite{douze2024faisslib}, the HNSW index of hnswlib~\cite{malkov2020hnsw} with 32 links per node and a
construction width of 200, and ScaNN~\cite{guo2020scann} with its anisotropic 4-bit quantization and exact
reordering run on all 192 cores of the second server. Finally, CertEx~\cite{zhao2026certex} returns
the exact top 10 over the whole dataset by certified exclusion on one GPU. The IVF indexes use the lists of
Table~\ref{tab:datasets} and sweep the number of probed lists from 1 to 500, or to 1,000 on the datasets of 100
million vectors, whereas HNSW sweeps its search width from 16 to 1,024.

\begin{table}[t]
  \centering
  \caption{The baselines and their answers.}
  \label{tab:baselines}
  \footnotesize
  \begin{tabular*}{\columnwidth}{@{\extracolsep{\fill}}lll@{}}
    \toprule
    Baseline & Index and hardware & Answer \\
    \midrule
    IVF-Flat (cuVS) & inverted lists, GPU & probed top 10 \\
    IVF-PQ + refine (cuVS) & quantized lists, GPU & approximate \\
    IVF-SQ + refine (cuVS) & quantized lists, GPU & approximate \\
    IVF-RaBitQ (cuVS) & RaBitQ lists, GPU & approximate \\
    CAGRA (cuVS) & proximity graph, GPU & approximate \\
    Brute force (cuVS) & exhaustive float32, GPU & exact (float32) \\
    IVF-Flat (Faiss) & inverted lists, CPU & probed top 10 \\
    HNSW (hnswlib) & proximity graph, CPU & approximate \\
    ScaNN & quantized partitions, CPU & approximate \\
    CertEx & certified exclusion, GPU & exact \\
    \bottomrule
  \end{tabular*}
\end{table}

\paragraph{Implementation.} \sysname\ is a CUDA extension on the GPU implementation of CertEx, whose kernels encode
the query rows with their certificates, run the routed tensor-core pass, and verify the survivors. The evaluated
instantiation uses the lifted width $m = 507$ on the 768- and 960-dimensional datasets, $m = 123$ on the SIFT
datasets, and $m = 91$ on the Deep1B prefixes, an anchor stride of $s = 16$, an anchor budget of $b = 1{,}024$, a threshold
slack of $\tau = 10^{-4}$, and batches of 8,192 queries. The verification computes the exact distances in float64 with one warp per query, in
rounds of at least 32 survivors. Float64 arithmetic costs little in this step, because a query verifies only a
few hundred vectors. On Deep-100M and DataComp, \sysname\ keeps the full vectors in pinned host memory.

\subsection{Passage Retrieval and Ranking Quality}
\label{sec:msmarco}

\begin{figure}[t]
  \centering
  \includegraphics[width=\columnwidth]{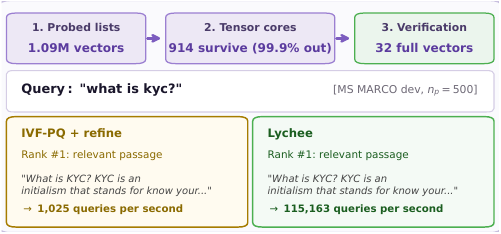}
  \caption{One MS MARCO dev query at 500 probed lists.}
  \label{fig:qual}
\end{figure}

Figure~\ref{fig:qual} traces one MS MARCO dev query through \sysname\ at 500 probed lists. The probed lists hold 1,086,390 passages, of which the certified screening on tensor cores keeps 914. The ordered verification then reads the full vectors of only 32 survivors before it returns the relevant passage at Rank~1. IVF-PQ with refinement returns the same passage at Rank~1, whereas its throughput at this setting is 1,025 queries per second against 115,163 for \sysname. Over all 6,980 dev queries at 500 probed lists, the two indexes reach an MRR@10 of 0.3334 and 0.3337 (Figure~\ref{fig:msmarco}), so that \sysname\ matches the ranking quality of refined quantized lists at 112 times the throughput of IVF-PQ.

\begin{figure}[t]
  \centering
  \includegraphics[width=\columnwidth]{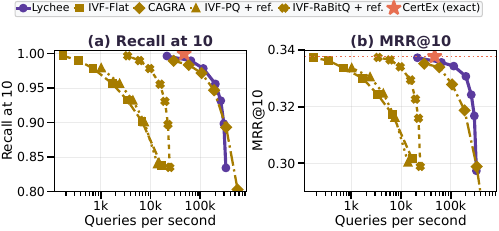}
  \caption{Passage recall and MRR@10 against throughput.}
  \label{fig:msmarco}
\end{figure}

Figure~\ref{fig:msmarco} plots recall at 10 and MRR@10 against throughput on MS MARCO. At 500 of 4,096 lists, recall at 10 reaches 0.979. \sysname\ answers 115,163 QPS at this operating point compared to 704 for IVF-Flat ($164\times$ faster), and 201,220 vs 1,713 QPS at 200 probed lists ($117\times$). Quantized indexes with full-precision refinement (IVF-PQ and IVF-SQ) lag by $70\times$ to $112\times$ (serving 1,025 and 1,650 QPS), because every query still scans the quantized codes of all its probed passages, although the refinement reads only $4k$ full vectors.

The right panel of Figure~\ref{fig:msmarco} evaluates the ranking quality of the retrieved passages. The float32 brute-force search of cuVS yields an MRR@10 of 0.3377 against official labels, slightly below the exact 0.3378, because its float32 distances order a few near-tied passages differently. \sysname\ preserves 0.3334 at 500 lists and 0.3357 at 1,000 lists. CAGRA explores fewer vectors at lower recall ($\le 0.95$), but saturates at recall 0.990 with 29,364 QPS and an MRR@10 of 0.3352. \sysname\ matches 0.990 recall at 65,178 QPS with higher MRR@10 (0.3357), delivering $2.2\times$ the throughput of CAGRA and $185\times$ that of IVF-Flat (352 QPS). Beyond recall 0.99 where probed lists cover half the corpus, exact search becomes more efficient: \sysname\ serves 21,457 QPS at recall 0.997, whereas exact CertEx retrieves the complete ground-truth passages with full MRR@10 (0.3378) at 47,630 QPS.

\subsection{Recall and Throughput}
\label{sec:recall}

\begin{figure}[t]
  \centering
  \includegraphics[width=\columnwidth]{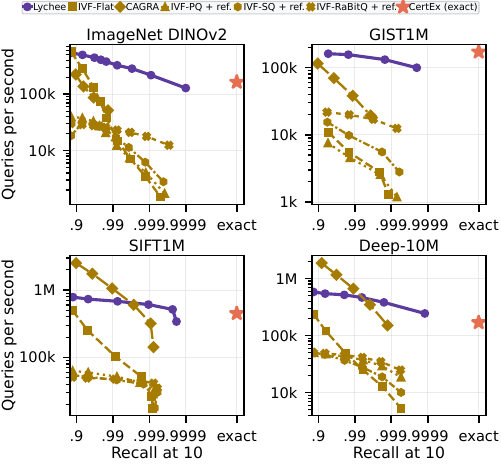}
  \caption{Recall at 10 and throughput on 1M--10M datasets.}
  \label{fig:recall-mid}
\end{figure}

Figure~\ref{fig:recall-mid} evaluates recall at 10 against throughput on ImageNet DINOv2, GIST1M, SIFT1M, and Deep-10M, testing whether certified screening preserves IVF recall while removing full-precision list scans. \sysname\ and IVF-Flat reach nearly the same recall at every $n_p$, since both return the exact top 10 of their probed lists (Proposition~\ref{prop:exact}) and differ only in their $k$-means lists. On wide embeddings, IVF-Flat throughput collapses as $n_p$ grows, whereas \sysname\ sustains high throughput: at $n_p = 200$, \sysname\ serves 217,195 QPS on DINOv2 and 131,302 on GIST1M at recall $\ge 0.998$ ($63\times$ and $47\times$ faster than IVF-Flat). Quantized indexes (IVF-SQ and IVF-RaBitQ with refinement) lag by $7.7\times$ to $35\times$, because every query still scans the codes of all its probed vectors. Graph-based CAGRA saturates at recall 0.986 on DINOv2 and 0.996 on GIST1M, where \sysname\ delivers $6.3\times$ and $6.8\times$ higher throughput at higher recall. On SIFT1M and Deep-10M, \sysname\ is $19.7\times$ and $30\times$ faster than IVF-Flat at $n_p = 200$; although CAGRA is faster at lower recall ($\le 0.995$) on SIFT1M, it saturates at recall 0.9992, where \sysname\ delivers $3.6\times$ higher throughput at recall 0.9998.

\begin{figure}[t]
  \centering
  \includegraphics[width=\columnwidth]{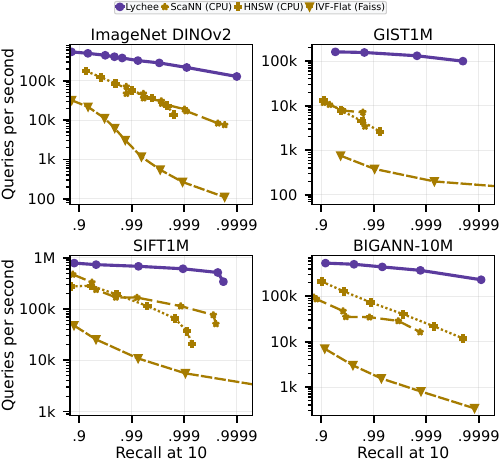}
  \caption{Throughput vs.\ CPU baselines on 192 cores.}
  \label{fig:recall-cpu}
\end{figure}

Figure~\ref{fig:recall-cpu} compares \sysname\ with CPU indexes running on all 192 cores of the second server. ScaNN, the strongest CPU baseline, reaches recall 0.999 at 18,866 QPS on DINOv2 and 105,628 on SIFT1M ($11.7\times$ and $5.7\times$ below \sysname), and saturates at recall 0.985 on GIST1M and 0.9987 on BIGANN-10M ($45\times$ and $23\times$ below \sysname). HNSW reaches at most recall 0.998 on DINOv2 and 0.992 on GIST1M, trailing \sysname\ by $14\times$ to $15\times$ at recall 0.999 on SIFT1M and BIGANN-10M. Faiss IVF-Flat serves 205 to 5,594 QPS at recall 0.999, which is $108\times$ to $894\times$ below \sysname.

\begin{figure}[t]
  \centering
  \includegraphics[width=\columnwidth]{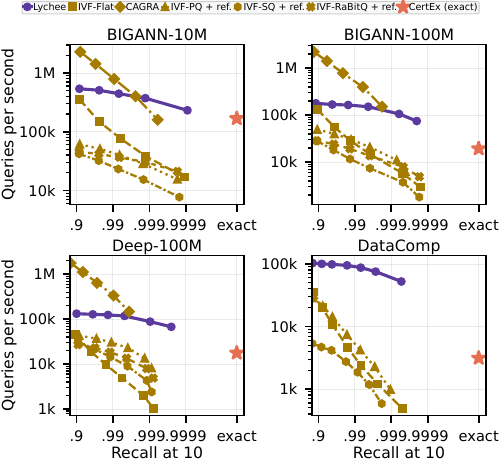}
  \caption{Recall at 10 and throughput on 10M--100M datasets.}
  \label{fig:recall-large}
\end{figure}

Figure~\ref{fig:recall-large} extends the evaluation to BIGANN-10M, BIGANN-100M, Deep-100M, and DataComp, where probed lists contain up to one million vectors per query. \sysname\ is $9.7\times$ to $10.4\times$ faster than IVF-Flat at $n_p = 200$, widening to $25\times$ on BIGANN-100M (75,053 vs 2,959 QPS at recall 0.9998) and $65\times$ on Deep-100M (67,548 vs 1,036 QPS at recall 0.9997) with full vectors in pinned host memory. Quantized indexes lag by $14\times$ to $42\times$, while CAGRA saturates at recall 0.998 on BIGANN-100M and 0.996 on Deep-100M. On DataComp, \sysname\ serves 75,565 QPS at recall 0.997 ($63\times$ faster than IVF-Flat, reaching $107\times$ at $n_p = 500$). At recall 0.999 on DataComp, \sysname\ is also $18\times$ faster than the exact search of CertEx, which answers 3,126 queries per second. Exact brute-force search is $36\times$ to $246\times$ slower than \sysname\ at recall 0.999 and suffers from float32 distance inaccuracies that return incorrect top-10 neighbors on Deep-10M and Deep-100M.

\subsection{Data Size Scalability}
\label{sec:size}

\begin{figure}[t]
  \centering
  \includegraphics[width=\columnwidth]{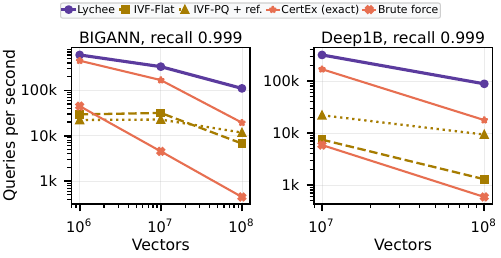}
  \caption{Throughput at recall 0.999 against dataset size.}
  \label{fig:size}
\end{figure}

Figure~\ref{fig:size} plots the throughput at recall 0.999 against the number of vectors on the BIGANN and Deep1B
prefixes, which tests the cost analysis of Section~\ref{sec:cost}. The exact methods screen or scan every vector, so
that their throughput falls in proportion to the dataset: CertEx answers 449,676, 168,264, and 19,427 queries per
second on SIFT1M, BIGANN-10M, and BIGANN-100M. \sysname\
instead screens only the probed lists, whose share of the dataset shrinks as the dataset grows, so that its throughput
at recall 0.999 reaches 1.3$\times$, 2.0$\times$, and 5.7$\times$ that of CertEx on the same three datasets.

The Deep1B prefixes show the same trend. \sysname\ is 1.9$\times$ faster than CertEx on Deep-10M and 5.0$\times$ faster on
Deep-100M, with the full vectors of Deep-100M in host memory. IVF-Flat, in contrast, stays 10$\times$ to 69$\times$ below \sysname\ at recall 0.999 on all
five prefixes, because it reads every vector of the probed lists at full precision. \sysname\ is therefore the
method of choice when a dataset is too large for exact screening of every row, while CertEx remains the exact
choice for datasets of a few million vectors.

\subsection{Billion-Scale Search on Multiple GPUs}
\label{sec:billion}

\begin{figure}[t]
  \centering
  \includegraphics[width=\columnwidth]{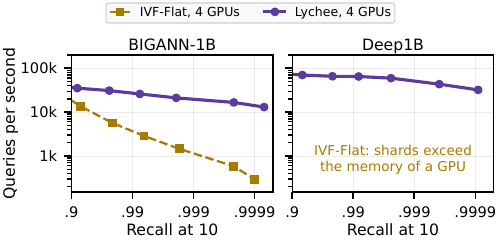}
  \caption{Billion-scale search on four H100 GPUs.}
  \label{fig:billion}
\end{figure}

Figure~\ref{fig:billion} evaluates one billion vectors on the four GPUs of the second server, which tests whether the
gain of \sysname\ holds at the scale of the largest public benchmarks~\cite{simhadri2022bigann}. Both methods split
the dataset into four contiguous shards of 250 million vectors with 16,384 $k$-means lists each, one shard per GPU. A query probes the $n_p$ nearest lists of every shard, every shard returns the exact top 10 of its probed
lists, and the host merges the four answers by distance. \sysname\ keeps the full vectors of its shard in pinned host
memory and only its lifted rows and anchors on its GPU, whereas IVF-Flat keeps its shard on the GPU.

On BIGANN-1B, \sysname\ answers 21,108 queries per second at recall 0.998 ($n_p = 200$), 16,615 at recall 0.9998
($n_p = 500$), and 13,032 at recall 0.9999 ($n_p = 1{,}000$), against 1,455, 583, and 293 for IVF-Flat, which is
14.5$\times$, 28.5$\times$, and 44.5$\times$ faster at equal recall. The build of \sysname\ takes 122~s against 398~s
for IVF-Flat. On Deep1B, a float32 shard of 250 million vectors occupies 96~GB and exceeds the memory of an H100, so
that IVF-Flat cannot build its index on the four GPUs. The lifted rows and anchors of \sysname\ occupy about 70~GB per
GPU, and \sysname\ answers 59,466 queries per
second at recall 0.998, 43,333 at recall 0.9996, and 32,082 at recall 0.9999 after a build of 270~s. \sysname\
therefore searches one billion float32 vectors on a single server at recall 0.9999, a configuration that the
full-precision lists of IVF-Flat cannot hold.

\subsection{Workload per Query}
\label{sec:work}

\begin{figure}[t]
  \centering
  \includegraphics[width=\columnwidth]{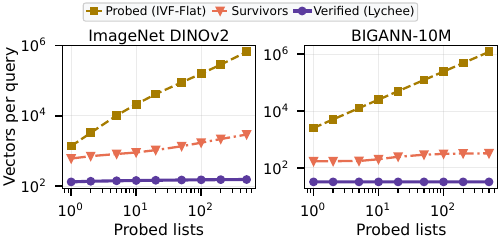}
  \caption{Probed, surviving, and verified vectors per query.}
  \label{fig:work}
\end{figure}

Figure~\ref{fig:work} counts the probed vectors, the survivors of the certified screening, and the verified vectors
per query as $n_p$ grows, which tests the cost model of Section~\ref{sec:cost}. On ImageNet DINOv2, the probed vectors
grow from 1,361 to 653,749 per query in proportion to $n_p$, while the survivors grow from 602 to 2,905 and the
verified vectors only from 131 to 152. On BIGANN-10M, the probed vectors grow to 1,206,157 per query, whereas the
survivors stay below 330 and the verification ends after its first round of 32 vectors. The lists that a larger $n_p$
adds lie far from the query, so that their vectors fail the certified test and never reach the verification. At
$n_p = 200$ on ImageNet DINOv2, \sysname\ therefore reads 151 full vectors per query against 282,650 for IVF-Flat, 1,870 times
fewer.

Survivors grow faster than verified vectors because the list-local threshold is looser than the final $k$-th distance. However, ordered verification (Algorithm~\ref{alg:search}) sorts survivors by lower bound and halts once a bound reaches the running $k$-th distance, reading 16.0$\times$ fewer vectors than verifying all survivors on ImageNet DINOv2 ($n_p = 200$, Table~\ref{tab:ablation}).

\subsection{Build Time and Host-Resident Vectors}
\label{sec:build}

\begin{figure}[t]
  \centering
  \includegraphics[width=\columnwidth]{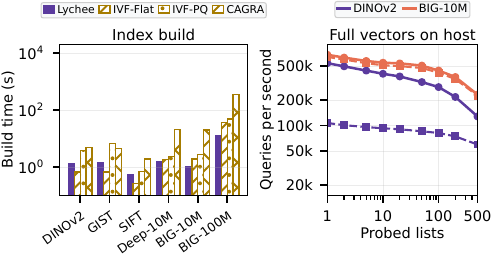}
  \caption{Index build time and host-resident throughput.}
  \label{fig:build}
\end{figure}

Figure~\ref{fig:build} reports the build time of every index and the throughput of the host-resident configuration, which test the cost of the index and the memory that it needs. The build time of \sysname\ covers the $k$-means training and assignment of its lists followed by one projection and encoding pass, building in 0.6 to 1.6~s on datasets up to 10M vectors and 12.4~s on BIGANN-100M. In contrast, CAGRA requires 1.9 to 20.3~s on smaller datasets and 362.7~s on BIGANN-100M ($29\times$ slower), while HNSW needs 10.2 to 185.8~s on 192 CPU cores. \sysname's index is therefore ready in seconds wherever a graph index takes minutes.

The host-resident configuration keeps full vectors in pinned host memory and only lifted rows and anchors on the GPU, reading survivors over PCIe. On ImageNet DINOv2, this uses 2.0~GB GPU memory (vs 3.9~GB for float32 vectors) while serving 74,627 QPS at recall 0.999 ($n_p = 200$, $22\times$ faster than device-resident IVF-Flat). On BIGANN-10M, it retains 95\% of device throughput at 200 probed lists. Device memory thus scales with $m$ rather than $D$ for wide embeddings, admitting large datasets like Deep-100M and DataComp onto a single GPU.

\subsection{Parameter Sensitivity}
\label{sec:sensitivity}

\begin{figure}[t]
  \centering
  \includegraphics[width=\columnwidth]{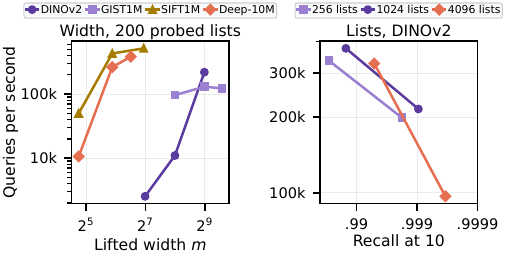}
  \caption{Parameter sensitivity to lifted width and lists.}
  \label{fig:sensitivity}
\end{figure}

Figure~\ref{fig:sensitivity} evaluates the two parameters of \sysname, the lifted width $m$ and the number of lists,
which test how tight the certified bound must be and how the partition trades list size against fineness. The left
panel plots the throughput at 200 probed lists against $m$. On ImageNet DINOv2, the throughput grows from 2,542
queries per second at $m = 127$ to 11,015 at $m = 255$ and 219,325 at $m = 507$, because the verified vectors per
query fall from 90,008 to 18,438 and 151. Deep-10M and SIFT1M follow the same trend up to 382,500 and 516,305 queries per second at their widest rows. GIST1M peaks at $m = 507$ (130,112 QPS), where $m = 763$ costs more than it saves. Wider rows pay off until the verified vectors approach the 32 of the first verification round.

The right panel plots recall against throughput for 256, 1,024, and 4,096 lists on ImageNet DINOv2 at the same shares
of probed lists. With 20\% of the lists probed, \sysname\ serves 199,023, 215,598, and 96,728 queries per second at
recall 0.9982, 0.9990, and 0.9997. Fewer and larger lists fill more of each tensor-core tile, because more queries of
a batch probe the same list, whereas more lists raise the recall at a given share. The default of 1,024 lists for a million vectors balances both effects, with throughput varying by $<2.3\times$ across a 16-fold list range.

\subsection{Ablation Study}
\label{sec:ablation}

\begin{table}[t]
  \centering
  \caption{Ablation of key components ($n_p = 200$).}
  \label{tab:ablation}
  \scriptsize
  \begin{tabular*}{\columnwidth}{@{\extracolsep{\fill}}llrrr@{}}
    \toprule
    Dataset & Variant & Survivors & Verified & QPS \\
    \midrule
    DINOv2   & \sysname\ (default)     & \textbf{2,410} & \textbf{151} & \textbf{223,688} \\
             & w/ global threshold     & 182,207        & 151          & 11,913 \\
             & w/o ordered verif.      & 2,410          & 2,410        & 34,812 \\
    \midrule
    SIFT1M   & \sysname\ (default)     & \textbf{208}   & \textbf{32}  & \textbf{575,043} \\
             & w/ global threshold     & 9,841          & 32           & 289,306 \\
             & w/o ordered verif.      & 208            & 208          & 214,500 \\
    \midrule
    BIGANN-10M & \sysname\ (default)   & \textbf{319}   & \textbf{32}  & \textbf{371,279} \\
             & w/ global threshold     & 81,930         & 32           & 49,235 \\
             & w/o ordered verif.      & 319            & 319          & 108,300 \\
    \bottomrule
  \end{tabular*}
\end{table}

Table~\ref{tab:ablation} isolates the two key components of \sysname, namely the list-local anchor threshold
(Section~\ref{sec:threshold}) and the ordered bound verification (Algorithm~\ref{alg:search}). The global variant
spreads the same budget of $b = 1{,}024$ anchors evenly over the whole dataset, as CertEx does. Such a threshold
does not guarantee $U(\vq) \ge r_k(\vq)$ (Section~\ref{sec:threshold}), although its anchors lie mostly far from the
query and make it much looser in practice than the $k$-th distance within the probed lists. Recall therefore stays at or above 0.999, whereas the survivors per query grow by
$75.6\times$, $47.3\times$, and $256.8\times$ on ImageNet DINOv2, SIFT1M, and BIGANN-10M. The throughput drops
accordingly by $18.8\times$, $2.0\times$, and $7.5\times$.

Without the ordered verification, every survivor is verified at full precision. The full-vector reads then grow by
$16.0\times$, $6.5\times$, and $10.0\times$ on the same datasets, so that the throughput drops by $6.4\times$ (from
223,688 to 34,812 QPS on ImageNet DINOv2), $2.7\times$, and $3.4\times$. Both components are therefore necessary,
since the threshold bounds the survivors, whereas the ordered verification bounds their full-precision reads.

\section{Conclusion}
\label{sec:conclusion}

High-recall search on inverted files spends its memory traffic on full-precision scans of the probed lists, while
graph indexes saturate below the high-recall targets. \sysname\ bounds the $k$-th distance within the probed lists
from their own anchors, screens the probed vectors on tensor cores with certified lower bounds, and verifies the few
survivors in increasing order of their bounds, which returns exactly the answer of IVF-Flat on the same lists. On eleven datasets of
up to one billion vectors and against ten GPU and CPU baselines, \sysname\ is up to 185$\times$ faster
than IVF-Flat at equal recall.
\bibliographystyle{abbrvnat}
\bibliography{refs}

@article{zhao2026certex,
  author       = {Dongfang Zhao},
  title        = {CertEx: Exact Visual Retrieval at Scale by Certified Exclusion},
  journal      = {arXiv preprint arXiv:2610.21325},
  year         = {2026},
  doi          = {10.48550/arXiv.2610.21325},
  url          = {https://arxiv.org/abs/2610.21325}
}

@inproceedings{ootomo2024cagra,
  author       = {Ootomo, Hiroyuki and Naruse, Akira and Nolet, Corey and Wang, Ray and Feher, Tamas and Wang, Yong},
  title        = {CAGRA: Highly Parallel Graph Construction and Approximate Nearest Neighbor Search for GPUs},
  booktitle    = {IEEE 40th International Conference on Data Engineering (ICDE)},
  pages        = {4236--4247},
  year         = {2024},
  doi          = {10.1109/ICDE60146.2024.00323},
  url          = {https://arxiv.org/abs/2308.15136}
}

@article{malkov2020hnsw,
  author       = {Malkov, Yu A. and Yashunin, D. A.},
  title        = {Efficient and Robust Approximate Nearest Neighbor Search Using Hierarchical Navigable Small World Graphs},
  journal      = {IEEE Transactions on Pattern Analysis and Machine Intelligence},
  volume       = {42},
  number       = {4},
  pages        = {824--836},
  year         = {2020},
  doi          = {10.1109/TPAMI.2018.2889473},
  url          = {https://doi.org/10.1109/TPAMI.2018.2889473}
}

@article{gao2024rabitq,
  author       = {Gao, Jianyang and Long, Cheng},
  title        = {RaBitQ: Quantizing High-Dimensional Vectors with a Theoretical Error Bound for Approximate Nearest Neighbor Search},
  journal      = {Proceedings of the ACM on Management of Data},
  volume       = {2},
  number       = {3},
  pages        = {1--27},
  year         = {2024},
  doi          = {10.1145/3654970},
  url          = {https://doi.org/10.1145/3654970}
}

@misc{nvidia2025cuvs,
  author       = {NVIDIA},
  title        = {{cuVS}: Vector Search and Clustering on the {GPU}},
  year         = {2026},
  howpublished = {\url{https://github.com/NVIDIA/cuvs}},
  note         = {Version 26.10}
}

@article{johnson2019faiss,
  author       = {Jeff Johnson and Matthijs Douze and Herv{\'e} J{\'e}gou},
  title        = {Billion-Scale Similarity Search with {GPUs}},
  journal      = {IEEE Transactions on Big Data},
  volume       = {7},
  number       = {3},
  pages        = {535--547},
  year         = {2019},
  doi          = {10.1109/TBDATA.2019.2921572}
}

@article{jegou2011pq,
  author       = {Jégou, Hervé and Douze, Matthijs and Schmid, Cordelia},
  title        = {Product Quantization for Nearest Neighbor Search},
  journal      = {IEEE Transactions on Pattern Analysis and Machine Intelligence},
  volume       = {33},
  number       = {1},
  pages        = {117--128},
  year         = {2011},
  doi          = {10.1109/TPAMI.2010.57},
  url          = {https://doi.org/10.1109/TPAMI.2010.57}
}

@inproceedings{sivic2003videogoogle,
  author       = {Josef Sivic and Andrew Zisserman},
  title        = {{Video Google}: A Text Retrieval Approach to Object Matching in Videos},
  booktitle    = {Proceedings of the IEEE International Conference on Computer Vision (ICCV)},
  pages        = {1470--1477},
  year         = {2003},
  doi          = {10.1109/ICCV.2003.1238663}
}

@article{aumuller2020annbench,
  author       = {Aumüller, Martin and Bernhardsson, Erik and Faithfull, Alexander},
  title        = {ANN-Benchmarks: A benchmarking tool for approximate nearest neighbor algorithms},
  journal      = {Information Systems},
  volume       = {87},
  pages        = {101374},
  year         = {2020},
  doi          = {10.1016/j.is.2019.02.006},
  url          = {https://doi.org/10.1016/j.is.2019.02.006}
}

@inproceedings{babenko2012imi,
  author       = {Artem Babenko and Victor Lempitsky},
  title        = {The Inverted Multi-Index},
  booktitle    = {Proceedings of the IEEE Conference on Computer Vision and Pattern Recognition (CVPR)},
  pages        = {3069--3076},
  year         = {2012},
  doi          = {10.1109/CVPR.2012.6248038},
  url          = {https://doi.org/10.1109/CVPR.2012.6248038}
}

@inproceedings{babenko2014additivequant,
  author       = {Artem Babenko and Victor Lempitsky},
  title        = {Additive Quantization for Extreme Vector Compression},
  booktitle    = {Proceedings of the IEEE Conference on Computer Vision and Pattern Recognition (CVPR)},
  pages        = {931--938},
  year         = {2014},
  doi          = {10.1109/CVPR.2014.124},
  url          = {https://openaccess.thecvf.com/content_cvpr_2014/papers/Babenko_Additive_Quantization_for_2014_CVPR_paper.pdf}
}

@inproceedings{babenko2016deep,
  author       = {Artem Babenko and Victor Lempitsky},
  title        = {Efficient Indexing of Billion-Scale Datasets of Deep Descriptors},
  booktitle    = {Proceedings of the IEEE Conference on Computer Vision and Pattern Recognition (CVPR)},
  pages        = {2055--2063},
  year         = {2016},
  doi          = {10.1109/CVPR.2016.226}
}

@inproceedings{baranchuk2018revisiting,
  author       = {Dmitry Baranchuk and Artem Babenko and Yury Malkov},
  title        = {Revisiting the Inverted Indices for Billion-Scale Approximate Nearest Neighbors},
  booktitle    = {Proceedings of the European Conference on Computer Vision (ECCV)},
  pages        = {209--224},
  year         = {2018},
  doi          = {10.1007/978-3-030-01258-8_13},
  url          = {https://openaccess.thecvf.com/content_ECCV_2018/html/Dmitry_Baranchuk_Revisiting_the_Inverted_ECCV_2018_paper.html}
}

@inproceedings{beyer1999meaningful,
  author       = {Beyer, Kevin and Goldstein, Jonathan and Ramakrishnan, Raghu and Shaft, Uri},
  title        = {When Is ``Nearest Neighbor'' Meaningful?},
  booktitle    = {Proceedings of the 7th International Conference on Database Theory (ICDT)},
  pages        = {217--235},
  year         = {1999},
  doi          = {10.1007/3-540-49257-7_15},
  url          = {https://doi.org/10.1007/3-540-49257-7_15}
}

@inproceedings{beygelzimer2006cover,
  author       = {Beygelzimer, Alina and Kakade, Sham and Langford, John},
  title        = {Cover trees for nearest neighbor},
  booktitle    = {Proceedings of the 23rd International Conference on Machine Learning (ICML)},
  pages        = {97--104},
  year         = {2006},
  doi          = {10.1145/1143844.1143857}
}

@inproceedings{chen2021spann,
  author       = {Qi Chen and Bing Zhao and Haidong Wang and Mingqin Li and Chuanjie Liu and Zengzhong Li and Mao Yang and Jingdong Wang},
  title        = {{SPANN}: Highly-efficient {Billion-scale} Approximate Nearest Neighbor Search},
  booktitle    = {Advances in Neural Information Processing Systems (NeurIPS)},
  volume       = {34},
  pages        = {5199--5212},
  year         = {2021},
  url          = {https://proceedings.neurips.cc/paper/2021/hash/299dc35e747eb77177d9cea10a802da2-Abstract.html}
}

@inproceedings{ciaccia1997mtree,
  author       = {Paolo Ciaccia and Marco Patella and Pavel Zezula},
  title        = {{M-tree}: An Efficient Access Method for Similarity Search in Metric Spaces},
  booktitle    = {Proceedings of the 23rd International Conference on Very Large Data Bases (VLDB)},
  pages        = {426--435},
  year         = {1997},
  url          = {http://www.vldb.org/conf/1997/P426.PDF}
}

@inproceedings{deng2009imagenet,
  author       = {Deng, Jia and Dong, Wei and Socher, Richard and Li, Li-Jia and Li, Kai and Fei-Fei, Li},
  title        = {{ImageNet}: A Large-Scale Hierarchical Image Database},
  booktitle    = {IEEE Conference on Computer Vision and Pattern Recognition (CVPR)},
  pages        = {248--255},
  year         = {2009},
  doi          = {10.1109/CVPR.2009.5206848}
}

@inproceedings{dong2011nndescent,
  author       = {Wei Dong and Moses Charikar and Kai Li},
  title        = {Efficient $k$-Nearest Neighbor Graph Construction for Generic Similarity Measures},
  booktitle    = {Proceedings of the 20th International Conference on World Wide Web (WWW)},
  pages        = {577--586},
  year         = {2011},
  doi          = {10.1145/1963405.1963487}
}

@article{douze2024faisslib,
  author       = {Douze, Matthijs and Guzhva, Alexandr and Deng, Chengqi and Johnson, Jeff and Szilvasy, Gergely and Mazaré, Pierre-Emmanuel and Lomeli, Maria and Hosseini, Lucas and Jégou, Hervé},
  title        = {The Faiss library},
  journal      = {IEEE Transactions on Big Data},
  volume       = {12},
  number       = {2},
  pages        = {346--361},
  year         = {2025},
  doi          = {10.1109/TBDATA.2025.3618474},
  url          = {https://doi.org/10.1109/TBDATA.2025.3618474}
}

@article{fu2019nsg,
  author       = {Cong Fu and Chao Xiang and Changxu Wang and Deng Cai},
  title        = {Fast Approximate Nearest Neighbor Search with the Navigating Spreading-out Graph},
  journal      = {Proceedings of the VLDB Endowment},
  volume       = {12},
  number       = {5},
  pages        = {461--474},
  year         = {2019},
  doi          = {10.14778/3303753.3303754}
}

@inproceedings{gadre2023datacomp,
  author       = {Gadre, Samir Yitzhak and Ilharco, Gabriel and Fang, Alex and Hayase, Jonathan and Smyrnis, Georgios and Nguyen, Thao and Marten, Ryan and Wortsman, Mitchell and Ghosh, Dhruba and Zhang, Jieyu and Orgad, Eyal and Entezari, Rahim and Daras, Giannis and Pratt, Sarah and Ramanujan, Vivek and Bitton, Yonatan and Marathe, Kalyani and Mussmann, Stephen and Vencu, Richard and Cherti, Mehdi and Krishna, Ranjay and Koh, Pang Wei and Saukh, Olga and Ratner, Alexander and Song, Shuran and Hajishirzi, Hannaneh and Farhadi, Ali and Beaumont, Romain and Oh, Sewoong and Dimakis, Alex and Jitsev, Jenia and Carmon, Yair and Shankar, Vaishaal and Schmidt, Ludwig},
  title        = {DataComp: In search of the next generation of multimodal datasets},
  booktitle    = {Advances in Neural Information Processing Systems (NeurIPS)},
  volume       = {36},
  pages        = {27092--27112},
  year         = {2023},
  doi          = {10.52202/075280-1179}
}

@article{gao2025extrabitq,
  author       = {Gao, Jianyang and Gou, Yutong and Xu, Yuexuan and Yang, Yongyi and Long, Cheng and Wong, Raymond Chi-Wing},
  title        = {Practical and Asymptotically Optimal Quantization of High-Dimensional Vectors in Euclidean Space for Approximate Nearest Neighbor Search},
  journal      = {Proceedings of the ACM on Management of Data},
  volume       = {3},
  number       = {3},
  pages        = {1--26},
  year         = {2025},
  doi          = {10.1145/3725413}
}

@inproceedings{garcia2008gpuknn,
  author       = {Garcia, Vincent and Debreuve, Eric and Barlaud, Michel},
  title        = {{Fast $k$-Nearest Neighbor Search Using {GPU}}},
  booktitle    = {2008 IEEE Computer Society Conference on Computer Vision and Pattern Recognition Workshops},
  pages        = {1--6},
  year         = {2008},
  doi          = {10.1109/cvprw.2008.4563100}
}

@article{ge2014opq,
  author       = {Tiezheng Ge and Kaiming He and Qifa Ke and Jian Sun},
  title        = {Optimized Product Quantization},
  journal      = {IEEE Transactions on Pattern Analysis and Machine Intelligence},
  volume       = {36},
  number       = {4},
  pages        = {744--755},
  year         = {2014},
  doi          = {10.1109/TPAMI.2013.240}
}

@article{groh2023ggnn,
  author       = {Groh, Fabian and Ruppert, Lukas and Wieschollek, Patrick and Lensch, Hendrik P. A.},
  title        = {GGNN: Graph-Based GPU Nearest Neighbor Search},
  journal      = {IEEE Transactions on Big Data},
  volume       = {9},
  number       = {1},
  pages        = {267--279},
  year         = {2023},
  doi          = {10.1109/tbdata.2022.3161156}
}

@inproceedings{guo2020scann,
  author       = {Guo, Ruiqi and Sun, Philip and Lindgren, Erik and Geng, Quan and Simcha, David and Chern, Felix and Kumar, Sanjiv},
  title        = {Accelerating Large-Scale Inference with Anisotropic Vector Quantization},
  booktitle    = {International Conference on Machine Learning (ICML)},
  series       = {Proceedings of Machine Learning Research},
  volume       = {119},
  pages        = {3887--3896},
  year         = {2020},
  url          = {https://proceedings.mlr.press/v119/guo20h.html}
}

@article{guo2022manu,
  author       = {Rentong Guo and Xiaofan Luan and Long Xiang and Xiao Yan and Xiaomeng Yi and Jigao Luo and Qianya Cheng and Weizhi Xu and Jiarui Luo and Frank Liu and Zhenshan Cao and Yanliang Qiao and Ting Wang and Bo Tang and Charles Xie},
  title        = {Manu: A Cloud Native Vector Database Management System},
  journal      = {Proceedings of the VLDB Endowment},
  volume       = {15},
  number       = {12},
  pages        = {3548--3561},
  year         = {2022},
  doi          = {10.14778/3554821.3554843}
}

@article{jagadish2005idistance,
  author       = {H. V. Jagadish and Beng Chin Ooi and Kian-Lee Tan and Cui Yu and Rui Zhang},
  title        = {{iDistance}: An Adaptive B+-Tree Based Indexing Method for Nearest Neighbor Search},
  journal      = {ACM Transactions on Database Systems (TODS)},
  volume       = {30},
  number       = {2},
  pages        = {364--397},
  year         = {2005},
  doi          = {10.1145/1071610.1071612}
}

@inproceedings{jegou2011bigann,
  author       = {J{\'e}gou, Herv{\'e} and Tavenard, Romain and Douze, Matthijs and Amsaleg, Laurent},
  title        = {Searching in one billion vectors: Re-rank with source coding},
  booktitle    = {2011 IEEE International Conference on Acoustics, Speech and Signal Processing (ICASSP)},
  pages        = {861--864},
  year         = {2011},
  doi          = {10.1109/ICASSP.2011.5946540},
  url          = {https://doi.org/10.1109/ICASSP.2011.5946540}
}

@inproceedings{karpukhin2020dpr,
  author       = {Vladimir Karpukhin and Barlas O{\u{g}}uz and Sewon Min and Patrick Lewis and Ledell Wu and Sergey Edunov and Danqi Chen and Wen-tau Yih},
  title        = {Dense Passage Retrieval for Open-Domain Question Answering},
  booktitle    = {Proceedings of the 2020 Conference on Empirical Methods in Natural Language Processing (EMNLP)},
  pages        = {6769--6781},
  year         = {2020},
  doi          = {10.18653/v1/2020.emnlp-main.550}
}

@inproceedings{khattab2020colbert,
  author       = {Omar Khattab and Matei Zaharia},
  title        = {{ColBERT}: Efficient and Effective Passage Search via Contextualized Late Interaction over {BERT}},
  booktitle    = {Proceedings of the 43rd International ACM SIGIR Conference on Research and Development in Information Retrieval},
  pages        = {39--48},
  year         = {2020},
  doi          = {10.1145/3397271.3401075}
}

@inproceedings{lewis2020rag,
  author       = {Patrick Lewis and Ethan Perez and Aleksandra Piktus and Fabio Petroni and Vladimir Karpukhin and Naman Goyal and Heinrich K{\"u}ttler and Mike Lewis and Wen-tau Yih and Tim Rockt{\"a}schel and Sebastian Riedel and Douwe Kiela},
  title        = {Retrieval-Augmented Generation for Knowledge-Intensive {NLP} Tasks},
  booktitle    = {Advances in Neural Information Processing Systems},
  volume       = {33},
  pages        = {9459--9474},
  year         = {2020},
  url          = {https://proceedings.neurips.cc/paper/2020/hash/6b493230205f780e1bc26945df7481e5-Abstract.html}
}

@article{lowe2004sift,
  author       = {Lowe, David G.},
  title        = {Distinctive Image Features from Scale-Invariant Keypoints},
  journal      = {International Journal of Computer Vision},
  volume       = {60},
  number       = {2},
  pages        = {91--110},
  year         = {2004},
  doi          = {10.1023/B:VISI.0000029664.99615.94},
  url          = {https://doi.org/10.1023/B:VISI.0000029664.99615.94}
}

@inproceedings{nguyen2016msmarco,
  author       = {Tri Nguyen and Mir Rosenberg and Xia Song and Jianfeng Gao and Saurabh Tiwary and Rangan Majumder and Li Deng},
  title        = {MS MARCO: A Human Generated MAchine Reading COmprehension Dataset},
  booktitle    = {Proceedings of the Workshop on Cognitive Computation: Integrating Neural and Symbolic Approaches, co-located with NeurIPS 2016},
  series       = {CEUR Workshop Proceedings},
  volume       = {1773},
  pages        = {Paper 9},
  numpages     = {5},
  year         = {2016},
  url          = {https://ceur-ws.org/Vol-1773/CoCoNIPS_2016_paper9.pdf}
}

@misc{nvidia2022h100,
  author       = {NVIDIA},
  title        = {{NVIDIA} {H100} Tensor Core {GPU} Architecture},
  year         = {2022},
  url          = {https://www.nvidia.com/en-us/data-center/h100/},
  howpublished = {Whitepaper}
}

@article{oquab2024dinov2,
  author       = {Oquab, Maxime and Darcet, Timothée and Moutakanni, Théo and Vo, Huy and Szafraniec, Marc and Khalidov, Vasil and Fernandez, Pierre and Haziza, Daniel and Massa, Francisco and El-Nouby, Alaaeldin and Assran, Mahmoud and Ballas, Nicolas and Galuba, Wojciech and Howes, Russell and Huang, Po-Yao and Li, Shang-Wen and Misra, Ishan and Rabbat, Michael and Sharma, Vasu and Synnaeve, Gabriel and Xu, Hu and Jegou, Hervé and Mairal, Julien and Labatut, Patrick and Joulin, Armand and Bojanowski, Piotr},
  title        = {DINOv2: Learning Robust Visual Features without Supervision},
  journal      = {Transactions on Machine Learning Research},
  year         = {2024},
  url          = {https://arxiv.org/abs/2304.07193}
}

@inproceedings{radford2021clip,
  author       = {Radford, Alec and Kim, Jong Wook and Hallacy, Chris and Ramesh, Aditya and Goh, Gabriel and Agarwal, Sandhini and Sastry, Girish and Askell, Amanda and Mishkin, Pamela and Clark, Jack and Krueger, Gretchen and Sutskever, Ilya},
  title        = {Learning Transferable Visual Models From Natural Language Supervision},
  booktitle    = {International Conference on Machine Learning (ICML)},
  series       = {Proceedings of Machine Learning Research},
  volume       = {139},
  pages        = {8748--8763},
  year         = {2021},
  url          = {https://proceedings.mlr.press/v139/radford21a.html}
}

@article{shi2026ivfrabitq,
  author       = {Jifan Shi and Jianyang Gao and James Xia and Tam{\'a}s B{\'e}la Feh{\'e}r and Cheng Long},
  title        = {{GPU}-Native Approximate Nearest Neighbor Search with {IVF-RaBitQ}: Fast Index Build and Search},
  journal      = {Proceedings of the VLDB Endowment},
  volume       = {19},
  number       = {11},
  pages        = {3454--3467},
  year         = {2026},
  doi          = {10.14778/3836663.3836701},
  url          = {https://doi.org/10.14778/3836663.3836701}
}

@inproceedings{simhadri2022bigann,
  author       = {Simhadri, Harsha Vardhan and Williams, George and Aumüller, Martin and Douze, Matthijs and Babenko, Artem and Baranchuk, Dmitry and Chen, Qi and Hosseini, Lucas and Krishnaswamy, Ravishankar and Srinivasa, Gopal and Subramanya, Suhas Jayaram and Wang, Jingdong},
  title        = {Results of the NeurIPS'21 Challenge on Billion-Scale Approximate Nearest Neighbor Search},
  booktitle    = {Proceedings of the NeurIPS 2021 Competitions and Demonstrations Track},
  series       = {Proceedings of Machine Learning Research},
  volume       = {176},
  pages        = {177--189},
  year         = {2022},
  url          = {https://proceedings.mlr.press/v176/simhadri22a.html}
}

@inproceedings{subramanya2019diskann,
  author       = {Subramanya, Suhas Jayaram and Devvrit and Kadekodi, Rohan and Krishnaswamy, Ravishankar and Simhadri, Harsha Vardhan},
  title        = {DiskANN: Fast Accurate Billion-point Nearest Neighbor Search on a Single Node},
  booktitle    = {Advances in Neural Information Processing Systems (NeurIPS)},
  volume       = {32},
  pages        = {13748--13758},
  year         = {2019},
  url          = {https://papers.nips.cc/paper_files/paper/2019/hash/09853c7fb1d3f8ee67a61b6bf4a7f8e6-Abstract.html}
}

@article{wang2021annsurvey,
  author       = {Mengzhao Wang and Xiaoliang Xu and Qiang Yue and Yuxiang Wang},
  title        = {A Comprehensive Survey and Experimental Comparison of Graph-Based Approximate Nearest Neighbor Search},
  journal      = {Proceedings of the VLDB Endowment},
  volume       = {14},
  number       = {11},
  pages        = {1964--1978},
  year         = {2021},
  doi          = {10.14778/3476249.3476255}
}

@inproceedings{wang2021milvus,
  author       = {Jianguo Wang and Xiaomeng Yi and Rentong Guo and Haihang Jin and Peng Xu and Shengjun Li and Xiangyu Wang and Xiangzhou Guo and Chengming Li and Xiaofan Xu and others},
  title        = {Milvus: A Purpose-Built Vector Data Management System},
  booktitle    = {Proceedings of the 2021 ACM SIGMOD International Conference on Management of Data},
  pages        = {2614--2627},
  year         = {2021},
  doi          = {10.1145/3448016.3457550}
}

@inproceedings{weber1998vafile,
  author       = {Roger Weber and Hans-J{\"o}rg Schek and Stephen Blott},
  title        = {A Quantitative Analysis and Performance Study for Similarity-Search Methods in High-Dimensional Spaces},
  booktitle    = {Proceedings of the 24th International Conference on Very Large Data Bases (VLDB)},
  pages        = {194--205},
  year         = {1998},
  url          = {http://www.vldb.org/conf/1998/p194.pdf}
}

@article{wei2020analyticdbv,
  author       = {Chuangxian Wei and Bin Wu and Sheng Wang and Renjie Lou and Chaoqun Zhan and Feifei Li and Yuanzhe Cai},
  title        = {{AnalyticDB-V}: A Hybrid Analytical Engine Towards Query Fusion for Structured and Unstructured Data},
  journal      = {Proceedings of the VLDB Endowment},
  volume       = {13},
  number       = {12},
  pages        = {3152--3165},
  year         = {2020},
  doi          = {10.14778/3415478.3415541}
}

@inproceedings{yianilos1993vptree,
  author       = {Peter N. Yianilos},
  title        = {{Data Structures and Algorithms for Nearest Neighbor Search in General Metric Spaces}},
  booktitle    = {Proceedings of the Fourth Annual ACM-SIAM Symposium on Discrete Algorithms (SODA)},
  pages        = {311--321},
  year         = {1993},
  url          = {https://dl.acm.org/doi/10.5555/313559.313789}
}

@inproceedings{zhao2020song,
  author       = {Zhao, Weijie and Tan, Shulong and Li, Ping},
  title        = {SONG: Approximate Nearest Neighbor Search on GPU},
  booktitle    = {2020 IEEE 36th International Conference on Data Engineering (ICDE)},
  pages        = {1033--1044},
  year         = {2020},
  doi          = {10.1109/icde48307.2020.00094}
}

@inproceedings{zhang2023vbase,
  author       = {Qianxi Zhang and Xuanhua Shi and Hanhua Chen and Hai Jin},
  title        = {{VBase}: Unifying High-Dimensional Vector and Relational Data Management},
  booktitle    = {Proceedings of the 17th USENIX Symposium on Operating Systems Design and Implementation (OSDI)},
  pages        = {421--438},
  year         = {2023},
  url          = {https://www.usenix.org/conference/osdi23/presentation/zhang-qianxi}
}

@inproceedings{hofstatter2021tasb,
  author       = {Sebastian Hofst{\"a}tter and Sheng-Chieh Lin and Jheng-Hong Yang and Jimmy Lin and Allan Hanbury},
  title        = {Efficiently Teaching an Effective Dense Retriever with Balanced Topic Aware Sampling},
  booktitle    = {Proceedings of the 44th International ACM SIGIR Conference on Research and Development in Information Retrieval},
  pages        = {113--122},
  year         = {2021},
  doi          = {10.1145/3404835.3462891},
  url          = {https://doi.org/10.1145/3404835.3462891}
}

@inproceedings{reimers2019sbert,
  author       = {Nils Reimers and Iryna Gurevych},
  title        = {Sentence-BERT: Sentence Embeddings using Siamese BERT-Networks},
  booktitle    = {Proceedings of the 2019 Conference on Empirical Methods in Natural Language Processing and the 9th International Joint Conference on Natural Language Processing (EMNLP-IJCNLP)},
  pages        = {3982--3992},
  year         = {2019},
  doi          = {10.18653/v1/D19-1410},
  url          = {https://aclanthology.org/D19-1410/}
}

@inproceedings{golub2024filtereddiskann,
  author       = {Siddharth Gollapudi and Neel Karia and Varun Sivashankar and Ravishankar Krishnaswamy and Nikit Begwani and Swapnil Raz and Yiyun Lin and Yin Zhang and Neelam Mahapatro and Premkumar Srinivasan and Amit Singh and Harsha Vardhan Simhadri},
  title        = {{Filtered-DiskANN}: Graph Algorithms for Approximate Nearest Neighbor Search with Filters},
  booktitle    = {Proceedings of the ACM Web Conference 2023 (WWW)},
  pages        = {3406--3416},
  year         = {2023},
  doi          = {10.1145/3543507.3583552},
  url          = {https://doi.org/10.1145/3543507.3583552}
}

@article{patel2024acorn,
  author       = {Liana Patel and Peter Kraft and Carlos Guestrin and Matei Zaharia},
  title        = {{ACORN}: Performant and Predicate-Agnostic Search Over Vector Embeddings and Structured Data},
  journal      = {Proceedings of the ACM on Management of Data},
  volume       = {2},
  number       = {3},
  pages        = {1--27},
  year         = {2024},
  doi          = {10.1145/3654923},
  url          = {https://doi.org/10.1145/3654923}
}

@article{wang2024starling,
  author       = {Mengzhao Wang and Weizhi Xu and Xiaomeng Yi and Songlin Wu and Zhangyang Peng and Xiangyu Ke and Yunjun Gao and Xiaoliang Xu and Rentong Guo and Charles Xie},
  title        = {{Starling}: An {I/O}-Efficient Disk-Resident Graph Index Framework for High-Dimensional Vector Similarity Search on Data Segment},
  journal      = {Proceedings of the ACM on Management of Data},
  volume       = {2},
  number       = {1},
  pages        = {1--27},
  year         = {2024},
  doi          = {10.1145/3639269},
  url          = {https://doi.org/10.1145/3639269}
}

@inproceedings{gui2026pilotann,
  author       = {Gui, Yuntao and Yin, Peiqi and Yan, Xiao and Zhang, Chaorui and Zhang, Weixi and Cheng, James},
  title        = {{PilotANN}: Memory-Bounded {GPU} Acceleration for Vector Search},
  booktitle    = {Proceedings of the 32nd ACM SIGKDD Conference on Knowledge Discovery and Data Mining},
  series       = {KDD '26},
  pages        = {348--358},
  numpages     = {11},
  year         = {2026},
  publisher    = {Association for Computing Machinery},
  address      = {New York, NY, USA},
  doi          = {10.1145/3770854.3780192},
  url          = {https://doi.org/10.1145/3770854.3780192},
  isbn         = {9798400722585},
  location     = {Republic of Korea}
}

@article{venkatasubba2025bang,
  author       = {Karthik Venkatasubba and Saim Khan and Somesh Singh and Harsha Vardhan Simhadri and Jyothi Vedurada},
  title        = {{BANG}: Billion-Scale Approximate Nearest Neighbour Search Using a Single {GPU}},
  journal      = {IEEE Transactions on Big Data},
  volume       = {11},
  number       = {6},
  pages        = {3142--3157},
  year         = {2025},
  doi          = {10.1109/TBDATA.2025.3581085},
  url          = {https://ieeexplore.ieee.org/document/11045134/}
}

@article{vidal1986aesa,
  author       = {Vidal Ruiz, Enrique},
  title        = {An algorithm for finding nearest neighbours in (approximately) constant average time},
  journal      = {Pattern Recognition Letters},
  volume       = {4},
  number       = {3},
  pages        = {145--157},
  year         = {1986},
  doi          = {10.1016/0167-8655(86)90013-9}
}

@inproceedings{faloutsos1994gemini,
  author       = {Faloutsos, Christos and Ranganathan, M. and Manolopoulos, Yannis},
  title        = {Fast subsequence matching in time-series databases},
  booktitle    = {Proceedings of the 1994 ACM SIGMOD international conference on Management of data},
  pages        = {419--429},
  year         = {1994},
  doi          = {10.1145/191839.191925},
  url          = {https://doi.org/10.1145/191839.191925}
}

@article{russakovsky2015ilsvrc,
  author       = {Russakovsky, Olga and Deng, Jia and Su, Hao and Krause, Jonathan and Satheesh, Sanjeev and Ma, Sean and Huang, Zhiheng and Karpathy, Andrej and Khosla, Aditya and Bernstein, Michael and Berg, Alexander C. and Fei-Fei, Li},
  title        = {ImageNet Large Scale Visual Recognition Challenge},
  journal      = {International Journal of Computer Vision},
  volume       = {115},
  number       = {3},
  pages        = {211--252},
  year         = {2015},
  doi          = {10.1007/s11263-015-0816-y}
}

@article{oliva2001gist,
  author       = {Oliva, Aude and Torralba, Antonio},
  title        = {Modeling the Shape of the Scene: A Holistic Representation of the Spatial Envelope},
  journal      = {International Journal of Computer Vision},
  volume       = {42},
  number       = {3},
  pages        = {145--175},
  year         = {2001},
  doi          = {10.1023/A:1011139631724},
  url          = {https://doi.org/10.1023/A:1011139631724}
}

@inproceedings{schroff2015facenet,
  author       = {Schroff, Florian and Kalenichenko, Dmitry and Philbin, James},
  title        = {{FaceNet}: A Unified Embedding for Face Recognition and Clustering},
  booktitle    = {Proceedings of the IEEE Conference on Computer Vision and Pattern Recognition (CVPR)},
  pages        = {815--823},
  year         = {2015},
  doi          = {10.1109/CVPR.2015.7298682}
}

@inproceedings{deng2019arcface,
  author       = {Deng, Jiankang and Guo, Jia and Xue, Niannan and Zafeiriou, Stefanos},
  title        = {{ArcFace}: Additive Angular Margin Loss for Deep Face Recognition},
  booktitle    = {Proceedings of the IEEE/CVF Conference on Computer Vision and Pattern Recognition (CVPR)},
  pages        = {4690--4699},
  year         = {2019},
  doi          = {10.1109/CVPR.2019.00482}
}

@article{muller2004cbir,
  author       = {M{\"u}ller, Henning and Michoux, Nicolas and Bandon, David and Geissbuhler, Antoine},
  title        = {A Review of Content-Based Image Retrieval Systems in Medical Applications: Clinical Benefits and Future Directions},
  journal      = {International Journal of Medical Informatics},
  volume       = {73},
  number       = {1},
  pages        = {1--23},
  year         = {2004},
  doi          = {10.1016/j.ijmedinf.2003.11.024}
}

@article{grossman2011tar,
  author       = {Grossman, Maura R. and Cormack, Gordon V.},
  title        = {Technology-Assisted Review in E-Discovery Can Be More Effective and More Efficient Than Exhaustive Manual Review},
  journal      = {Richmond Journal of Law and Technology},
  volume       = {17},
  number       = {3},
  pages        = {1--48},
  year         = {2011},
  url          = {https://scholarship.richmond.edu/jolt/vol17/iss3/5/}
}

@inproceedings{lee2022dedup,
  author       = {Lee, Katherine and Ippolito, Daphne and Nystrom, Andrew and Zhang, Chiyuan and Eck, Douglas and Callison-Burch, Chris and Carlini, Nicholas},
  title        = {Deduplicating Training Data Makes Language Models Better},
  booktitle    = {Proceedings of the 60th Annual Meeting of the Association for Computational Linguistics (ACL)},
  pages        = {8424--8445},
  year         = {2022},
  doi          = {10.18653/v1/2022.acl-long.577}
}

@article{abbas2023semdedup,
  author       = {Abbas, Amro and Tirumala, Kushal and Simig, D{\'a}niel and Ganguli, Surya and Morcos, Ari S.},
  title        = {{SemDeDup}: Data-Efficient Learning at Web-Scale through Semantic Deduplication},
  journal      = {arXiv preprint arXiv:2303.09540},
  year         = {2023},
  doi          = {10.48550/arXiv.2303.09540},
  url          = {https://arxiv.org/abs/2303.09540}
}

\end{document}